\documentclass[10pt]{article}

\usepackage[normalem]{ulem}
\usepackage{mathtools}

\usepackage{amsmath}
\usepackage{amssymb}
\usepackage{amsfonts}
\usepackage{amsthm}
\usepackage{a4wide}
\usepackage{graphicx}
\usepackage{color}
\usepackage[sans]{dsfont}
\usepackage{pgfkeys}
\usepackage{longtable}
\usepackage{pdflscape}
\usepackage{float}
\usepackage{cancel}
\usepackage{dsfont}
\usepackage{setspace}
\usepackage{array}
\usepackage{hyperref}
\usepackage[english]{babel}
\usepackage[T1]{fontenc}
\usepackage{tikz,pgflibraryshapes}
\usepackage{enumerate}
\usetikzlibrary{decorations}
\usetikzlibrary{decorations.pathmorphing}
    \usetikzlibrary{decorations.pathreplacing}
    \usetikzlibrary{decorations.shapes}
    \usetikzlibrary{decorations.text}
    \usetikzlibrary{decorations.markings}
    \usetikzlibrary{decorations.fractals}
    \usetikzlibrary{decorations.footprints}

\usepackage{wasysym}
\usepackage{authblk}
\usepackage{calrsfs}

\newcommand{\id}{\mathds{1}}

\theoremstyle{plain}

\newtheorem{theorem}{Theorem}[section]
\newtheorem{lemma}[theorem]{Lemma}%[subsection]
\newtheorem{corollary}[theorem]{Corollary}%[subsection]
\newtheorem{definition}[theorem]{Definition}%[subsection]

\newtheorem{assumption}[theorem]{Assumption}

\usepackage{fancyhdr}

\DeclareMathAlphabet{\pazocal}{OMS}{zplm}{m}{n}

\numberwithin{equation}{section}

\title{  Indirect acquisition of information in quantum mechanics: states associated with tail events }

\begin{document}

%\nocite{*} Verify References

\author[1]{ \small M. Ballesteros\thanks{miguel.ballesteros@iimas.unam.mx}}
\author[2]{M. Fraas\thanks{martin.fraas@gmail.com}}
\author[3]{J. Fr\"ohlich\thanks{juerg@phys.ethz.ch}}
\author[4]{B. Schubnel\thanks{baptiste.schubnel@yahoo.fr}}
\affil[1]{Department of Mathematical Physics, Applied Mathematics and Systems Research Institute (IIMAS),  National Autonomous University of Mexico (UNAM)}
\affil[2]{Instituut voor Theoretische Fysica, KU Leuven}
\affil[3]{Institut f{\"u}r Theoretische Physik,   ETH Zurich}
\affil[4]{Swiss Federal Railways (SBB)}

\renewcommand\Authands{ and }

\normalsize 
\date \today

\maketitle

\begin{center}
Dedicated to the memory of Rudolf Haag
\end{center}

\begin{abstract}
The problem of reconstructing information on a physical system from data acquired in long sequences of direct (projective) measurements of some simple physical quantities - {\it histories} -  is analyzed within quantum mechanics; that is, the quantum theory of indirect measurements, and, in particular, of \textit{non-demolition} measurements is studied. It is shown that indirect measurements of \textit{time-independent} features of physical systems can be described in terms of quantum-mechanical operators belonging to an {\it  ``algebra of asymptotic observables''}. Our proof involves associating a natural measure space with certain sets of histories of a system and showing that quantum-mechanical states of the system determine probability measures on this space. Our main result then says that functions on that space of histories \textit{measurable at infinity} (i.e., functions that only depend on the ``tails'' of histories) correspond to operators in the algebra of asymptotic observables.

\end{abstract}

\section{Introduction}

This paper is a contribution to the foundations of quantum mechanics, in particular to the theory of \textit{indirect (non-demolition) measurements}. We study the problem of reconstructing information on a physical system, $S$, from data obtained in long sequences of direct (or projective) measurements of just a few physical quantities characteristic of $S$. Such sequences of outcomes of direct measurements are called {\it  histories}, and information reconstructed from a history is called an \textit{``indirect measurement''}. If the information reconstructed from histories is time-independent one speaks of an (indirect) ``non-demolition measurement''.

In this paper, we further elaborate on concepts introduced in \cite{BaFrFrSc} and prove some new results on non-demolition measurements. We will formulate our insights and results within the so-called  ``{\it ETH approach}'' to quantum mechanics, and ``$ETH$'' is an abbreviation for `` Events,  Trees and Histories''. (This approach has been introduced in \cite{Juerg}, \cite{BlFrSc}, \cite{FS}. The work presented in the following sections is a continuation of efforts described in these papers and in \cite{BaFrFrSc}.)

The theory of indirect and, in particular, of indirect non-demolition measurements \cite{Kraus} has recently attracted considerable attention; see \cite{Mass}, \cite{BaBe}, (and references given there). Beautiful experimental advances connected to the theoretical issues discussed in our paper have been made in the groups of S. Haroche and D. J. Wineland, which earned them the 2012 Nobel Prize in Physics. For example, the Haroche-Raimond group has succeeded in indirectly determining the number of photons in a cavity from long sequences of direct measurements of passive probes sent through the cavity, one after another; see \cite{guerlin}.  The probes are Rydberg atoms, specifically Rubidium atoms prepared in a coherent superposition of two highly excited internal states. When these atoms stream through the cavity their state precesses in the subspace spanned by those two states at an angular velocity depending on the number of photons in the cavity. A measurement  of the projection of the internal state of each probe onto a fixed state (i.e., a ``polarization measurement'') is made right after it has emerged from the cavity. It turns out that the statistics of the resulting measurement data determines the number of photons (modulo some integer $N<\infty$) present in the cavity uniquely; (relative phases between states corresponding to different photon numbers become unobservable). These experiments provide an example of the phenomenon of ``purification of states'' described in \cite{Mass}. A gedanken experiment with very similar features, where the cavity is replaced by a quantum dot capturing electrons, is described in Section \ref{SE}. 

Next, we sketch the main objectives of this paper. We take for granted the theory of repeated direct measurements in quantum mechanics, as described in \cite{Griffiths}, \cite{BlFrSc}, and references given there. This theory provides a theoretical underpinning for a formula, originally proposed by L\"uders, Schwinger and Wigner (the ``LSW formula''), enabling one to assign a \textit{probability} to every (finite) ``history'' of direct measurement outcomes. The LSW formula together with a basic hypothesis of \textit{``decoherence''} play an important role in our analysis. Our goal is to clarify what kind of information about a physical system can be retrieved from the asymptotics of very long histories of direct measurement outcomes, i.e., from tails of histories. This is a problem in the theory of indirect measurements in quantum mechanics, as first developed in \cite{Kraus}; see also \cite{Holevo}, \cite{Mass}, \cite{BaBe}, \cite{BaBeBe2}, \cite{BaBeTi}, \cite{BePel}, \cite{BaFrFrSc}, and references given there. We will show that, under appropriate hypotheses, tails of histories are in a $1$ - $1$ correspondence with eigenvalues of \textit{time-independent} observables of the system that generate a commutative algebra. The decomposition of probabilities of histories into convex combinations of extremal probabilities predicting 
``$0$ - $1$ laws'' for the tails of histories (i.e., a ``disintegration formula'' for probability measures on the space of histories) will be shown to be determined by the spectral decomposition of those time-independent observables.
%Iterated measurements produce information loss in such a way that the ``amount'' of observables at infinity (the self adjoint operators representing quantities one can observe after infinite sequences of direct measurements) is much smaller than the ``amount'' of observables (physical quantities one can- potentially -observe before the occurrence of any direct measurement). Clearly, our analysis  must rely on an statistical study of the histories one can obtain in certain experimental situations. 

Our results are formulated mathematically in Theorems   \ref{rep}, \ref{Este} and   \ref{decom} of Section \ref{NI}, which summarize the main new insights contained in our paper. An example of a simple physical system that can be analyzed with the help of our general results is described in Section \ref{SE}. Proofs of our main theorems are presented in Section \ref{Proofs}.  In Appendix \ref{appendix}  we recall some basic results from the theory of operator algebras.

\subsection{The general framework}\label{GF}

As announced above, the purpose of this paper is to make a contribution to the theory of indirect non-demolition measurements. A necessary prerequisite for a theory of indirect measurements is a theory of direct (``projective'') measurements in quantum mechanics. There are numerous, partly contradictory proposals of how to set up a quantum theory of direct measurements or observations. In this paper we 
follow the so-called ``{\it ETH approach}'' to quantum theory, as sketched in \cite{FS}, \cite{BlFrSc} and \cite{Juerg}. Here we briefly summarize those ingredients of this approach that will be needed later.

%In contrast to \cite{BlFrSc}, here we consider only non-demolition measurements, i.e., we neglect interference terms that measurements produce in forthcoming measurements and we assume that every two events happening in very distant times are independent to each other (see \eqref{deco} and the lines below \eqref{abelian} for a precise definition). This restriction permits proving several results in a relatively simple fashion, which would very difficult to achieve in a general setting. In this text we report our investigations presented in \cite{BaFrFrSc} and, moreover, we further develop the theory and prove new results. Immediately, we introduce - more precisely - the physical situation we analyze.       
In the ``ETH approach'' to quantum theory, a physical system $S$ is characterized by the following data;
(definitions of mathematical objects such as $C^{*}$- and von Neumann algebras, of states on operator algebras, and a sketch of the so-called Gel'fand-Naimark-Segal construction are presented in Appendix A).
\begin{itemize}
\item[I.]{Directly measurable physical quantities or properties (sometimes called \textit{``observables''}) of $S$ are represented by bounded self-adjoint linear operators that generate a $C^{*}$-algebra $\mathcal{E}\equiv \mathcal{E}_{S}$; i.e., an algebra of bounded operators invariant under taking adjoints, $^{*}$, and closed in norm, $\Vert \cdot \Vert$. (The norm of a $C^{*}$-algebra has properties we know of the norm of bounded operators acting on a Hilbert space. It is often convenient to view $\mathcal{E}$ as an abstract algebra not necessarily described in terms of bounded operators acting on a Hilbert space -- see, e.g., \cite{BrRo}, and Appendix A.) For simplicity, we suppose that the spectra of all operators in $\mathcal{E}$
corresponding to physical quantities/observables of $S$ are \textit{finite point spectra}. Then an operator $A=A^{*}$ representing a physical quantity of $S$ has a spectral decomposition
\begin{equation}\label{specdec}
A=\sum_{\alpha \in \sigma(A)} \alpha\, \Pi_{\alpha},
\end{equation}
where $\sigma(A)$ is the spectrum of $A$, i.e., the set of eigenvalues of the operator $A=A^{*}$, and 
$\Pi_{\alpha}\equiv \Pi_{\alpha}^{(A)}$ is the spectral projection of $A$ corresponding to the eigenvalue 
$\alpha \in \sigma(A)$ of $A$.
}
\item[II.]{An \textit{event} potentially detectable in $S$ corresponds to an orthogonal projection $\Pi=\Pi^{*}$ in the algebra $\mathcal{E}$. But \textit{not all} orthogonal projections in $\mathcal{E}$ represent events: projections corresponding to potentially detectable events are spectral projections of selfadjoint operators in $\mathcal{E}$ that represent physical quantities/observables of $S$.
}
\item[III.]{Physical time is a real parameter and is denoted by $t\in \mathbb{R}$. All physical quantities of $S$ potentially measurable/observable during the time interval $[s,t] \subset \mathbb{R}$ generate an algebra
$\mathcal{E}_{[s,t]}$.
It is natural to assume that if $[s',t'] \subset [s,t]$, i.e., $s\leq s', t\geq t'$, then
\begin{equation}\label{nesting}
\mathcal{E}_{[s',t']} \subseteq \mathcal{E}_{[s,t]} \subseteq \mathcal{E}.
\end{equation}
This property allows us to define algebras
\begin{equation}\label{after t}
\mathcal{E}_{\geq t} := \overline{\bigvee_{t':t<t'<\infty} \mathcal{E}_{[t,t']}}, \qquad \text{for  }t\in \mathbb{R},
\end{equation}
where the closure is taken to be in the operator norm of $\mathcal{E}$. The algebra $\mathcal{E}_{\geq t}$ is the algebra generated by all events potentially detectable at times $\geq t$. By property \eqref{nesting}, we have that 
\begin{equation}\label{telescope}
\mathcal{E} \supseteq \mathcal{E}_{\geq t} \supseteq \mathcal{E}_{\geq t'} \supset \mathcal{E}_{[t',t'']},
\end{equation}
whenever $t<t' <t'' (< \infty)$.
}
\item[IV.]{An \textit{``instrument''} of $S$ is an abstract abelian $C^{*}$- algebra, $\mathcal{I}$, of operators. For the purposes of this paper, we may assume that every instrument $\mathcal{I}$ is generated by a \textit{finite family of commuting orthogonal projections}, $\lbrace \Pi_{\xi} \vert \xi \in \mathcal{X} \rbrace$, where $\mathcal{X}$ is a finite set.\footnote{more generally, $\mathcal{X}$ is assumed to be a compact ``Polish space''} Typically, $\mathcal{I}$ may be thought to be generated by a family of commuting selfadjoint operators representing abstract physical quantities of $S$. It is an integral part of the definition of a system $S$ to specify the list, 
\begin{equation}\label{instrument}
\mathcal{O}_{S}:=\lbrace \mathcal{I}_{\ell} \rbrace_{\ell \in \mathcal{L}_{S}},
\end{equation}
of all instruments of $S$, where $\mathcal{L}_{S}$ labels the different instruments of $S$.\\
We remark that our notion of ``instruments'' deviates from more standard notions used in the literature. However, it is our notion that will be useful for the purposes of this paper.\\
We assume that, given a time $t \in \mathbb{R}$ and an arbitrary instrument 
$\mathcal{I} = \lbrace \Pi_{\xi}\vert \xi \in \mathcal{X} \rbrace \in \mathcal{O}_{S}$, there exists a representation of $\mathcal{I}$ in the algebra 
$\mathcal{E}_{\geq t}$, 
\begin{equation}\label{reps}
\mathcal{I} \ni \Pi_{\xi} \mapsto \Pi_{\xi}(t) \in \mathcal{E}_{\geq t},
\end{equation}
with $\Pi_{\xi}(t)=\Pi_{\xi}(t)^{*}$ and $\Pi_{\xi}(t)^{2}=\Pi_{\xi}(t)$, for all $\xi \in \mathcal{X}$.
}
\end{itemize}

\textit{States} of the system $S$ are normalized, positive linear functionals on the algebra $\mathcal{E}$; i.e., a state, $\omega$, of $S$ associates with every operator $A\in \mathcal{E}$ its expectation value, 
$\omega(A)$, in $\omega$, with the properties
\begin{itemize}
\item{$\omega$ is a linear functional on $\mathcal{E}$} with values in the complex numbers
\item{$\omega(A) \geq 0$ whenever $A >0$}
\item{$\underset{\lbrace A\in \mathcal{E}\vert \Vert A \Vert \leq 1\rbrace}{\text{sup}} \vert \omega(A) \vert = 1$ }
\end{itemize}
Given a state $\omega$ on $\mathcal{E}$, there exists a Hilbert space $\mathcal{H}_{\omega}$, a unit vector $\Omega \in \mathcal{H}_{\omega}$, and a representation, $\pi_{\omega}$, of the algebra 
$\mathcal{E}$ on $\mathcal{H}_{\omega}$ such that
\begin{equation}\label{GNS}
\omega(A)=\langle \Omega, \pi_{\omega}(A) \Omega \rangle_{\mathcal{H}_{\omega}},
\end{equation}
where $\langle \cdot, \cdot \rangle_{\mathcal{H}_{\omega}}$ is the scalar product on 
$\mathcal{H}_{\omega}$.
This is the content of the so-called Gel'fand-Naimark-Segal construction; (see Appendix A and \cite{BrRo}).

In the following, we may always assume that the algebra $\mathcal{E}$ is represented as an algebra of bounded linear operators acting on a separable Hilbert space $\mathcal{H}_{S}$,\footnote{The space $\mathcal{H}_{S}$ may be thought of as the direct sum of the GNS Hilbert spaces $\mathcal{H}_{\omega}$ over ``physically relevant'' states $\omega$ on $\mathcal{E}$}
 and that all physically relevant states $\omega$ on $\mathcal{E}$ are given by \textit{density matrices}, $P_{\omega}$, on $\mathcal{H}_{S}$ (i.e., positive operators on $\mathcal{H}_{S}$ with trace $=1$), with 
 $$\omega(A)=\text{tr}(P_{\omega}\,A), \quad A \in \mathcal{E}.$$

We then define the von Neumann algebras
\begin{eqnarray}\label{vNalg}
\mathcal{E}^{-} & := & \overline{\mathcal{E}}^{w}, \nonumber \\
\mathcal{E}_{\geq t}^{-} & := & \overline{\mathcal{E}_{\geq t}}^{w},
\end{eqnarray}
where the closures on the right sides of \eqref{vNalg} are taken in the weak operator topology defined on the algebra, $B(\mathcal{H}_{S})$, of all bounded operators on the Hilbert space 
$\mathcal{H}_{S}$. Then 
$$ \mathcal{E}_{\geq t}^{-} \subseteq \mathcal{E}^{-} \subseteq B(\mathcal{H}_{S}), \qquad \forall \text{ times  }t \in \mathbb{R}.$$
The following notions all depend on our choice of the representation of $\mathcal{E}$ on $\mathcal{H}_{S}$.

\textit{Definition} (``Algebra of asymptotic observables'').
\begin{equation}\label{obs-at-infty}
\mathcal{E}_{\infty}:= \bigcap_{t \in \mathbb{R}} \mathcal{E}_{\geq t}^{-}.
\end{equation}

A key property in the theory of direct (projective) measurements is the following property of \textit{``Loss of Access to Information''}: We consider systems $S$ (isolated, but \textit{open}!) with the property that
\begin{equation}\label{LOI}
\boxed{ \mathcal{E}_{\geq t}^{-} \underset{\not=}{\supset} \mathcal{E}_{\geq t'}^{-}, \quad \text{ for  } t'>t}
\end{equation}

The following remarks are intended to illuminate the importance of Property \eqref{LOI}; (but they do not cover some of the more subtle aspects of the theory of direct measurements). Because of \eqref{LOI} it can and, in general, will happen that the state $\omega_{t}$ on the algebra $\mathcal{E}_{\geq t}^{-}$, defined by
\begin{equation}\label{restr-state}
\omega_{t}:= \omega \vert_{\mathcal{E}_{\geq t}^{-}}, \quad \omega_{t}(A) = \omega(A), \quad \forall A \in \mathcal{E}_{\geq t}^{-},
\end{equation}
is a \textit{mixed} state \textit{even} if $\omega$ is \textit{pure} as a state on the algebra $\mathcal{E}$. In that case, there are states $\omega^{(1)}$ and $\omega^{(2)}$ on the algebra $\mathcal{E}_{\geq t}^{-}$, with 
$\omega^{(1)}\not= \omega^{(2)}$, and a positive number $\lambda < 1$ such that
$$\omega(A)=\lambda \omega^{(1)}(A) + (1-\lambda) \omega^{(2)}(A), \quad \forall A \in \mathcal{E}_{\geq t}^{-}.$$

Under very general hypotheses, one can determine the transition probability for a transition from the state 
$\omega_{t}$ on the algebra $\mathcal{E}_{\geq t}^{-}$ to the state $\omega^{(1)}$. Let us denote this probability by $p(\omega_{t}, \omega^{(1)})$. Given that the system has been prepared in state $\omega$ at some time \textit{earlier} than $t$, it will happen with probability $p(\omega_{t}, \omega^{(1)})$ that the state 
$\omega^{(1)}$ yields more accurate predictions about the behavior of the system $S$ at times $\geq t$. (The optimal choice of the state $\omega^{(1)}$ depends on the events happening in the system $S$ after the preparation of $S$ in the state $\omega$ and before time $t$.) Assuming that we have reasons to expect that $\omega^{(1)}$ is the state that describes the behavior of $S$ at times $\geq t$ most accurately, we may ask what this state predicts about the appearance of events in $S$ at times $\geq t$.
Well, there could exist an instrument \mbox{$\mathcal{I}=\lbrace \Pi_{\xi} \vert \xi \in \mathcal{X} \rbrace \in \mathcal{O}_{S}$} such that 
\begin{equation}\label{event}
\omega^{(1)}(A)= \sum_{\xi \in \mathcal{X}}\omega^{(1)}(\Pi_{\xi}(t)\, A \, \Pi_{\xi}(t)), \quad \forall A \in 
\mathcal{E}_{\geq t}^{-},
\end{equation}
up to possibly a tiny error. Under certain conditions on the projections $\Pi_{\xi}(t), \xi \in \mathcal{X}$, that we do not wish to repeat here\footnote{namely that these operators belong to (or are very close in norm to) the ``centre of the centralizer'' of the state 
$\omega^{(1)}$ with respect to the algebra $\mathcal{E}_{\geq t}^{-}$, see \cite{FS, Juerg}} it then follows that if $S$ is prepared in the state $\omega^{(1)}$ then, around time $t$, the instrument $\mathcal{I}$ is triggered and displays the value $\xi$ with a probability given by \textit{Born's Rule}, i.e.,
\begin{equation}\label{Born}
\text{Prob}_{\omega^{(1)}}\lbrace \xi \rbrace = \omega^{(1)}(\Pi_{\xi}(t)), \quad \forall \xi \in \mathcal{X}.
\end{equation}
Let us assume that the system $S$ has been prepared in state $\omega$ some time before an event is detected around time $t$. Assuming that the instrument $\mathcal{I}$ would \textit{not} be triggered around time $t$ if $S$ had been prepared in state 
$\omega^{(2)}$ then the probability that $\mathcal{I}$ is not triggered, given that $S$ has been prepared in state $\omega$, is given by $1-p(\omega_{t}, \omega^{(1)})$, while the probability that $\mathcal{I}$ is triggered around time $t$ and displays the value 
$\xi$ is given by 
\begin{equation}\label{Born's Rule}
\text{Prob}\lbrace \text{event  }\Pi_{\xi} \text{  appears around time  } t \rbrace = p(\omega_{t}, \omega^{(1)}) \cdot \omega^{(1)}\big(\Pi_{\xi}(t)\big).
\end{equation}
 For a more complete exposition of the theory of direct measurements in the ``ETH approach'' to quantum mechanics we refer the reader to \cite{FS, BlFrSc, Juerg}.
In this paper, we consider a very simple special case of the general theory, which relies on the following idealizations.

\textit{A simple model of a physical system} $S$:

We consider a quantum-mechanical model of a physical system $S$ exhibiting the following features: $S$ has a single instrument $\mathcal{I}=\lbrace \Pi_{\xi} \vert \xi \in \mathcal{X} \rbrace, \text{card}\mathcal{X} < \infty$, with
\begin{equation}\label{partition of 1}
\sum_{\xi \in \mathcal{X}} \Pi_{\xi}= {\bf{1}}.
\end{equation}
 Furthermore, we assume that, for a large family of physically interesting states of $S$,
the instrument $\mathcal{I}$ is triggerd at infinitely many times $t_1 < t_2 < \cdots < t_n < \cdots$; (this sequence of times might, of course, depend on the initial state of $S$). Whenever $\mathcal{I}$ is triggered it displays a value $\xi \in \mathcal{X}$, in the sense of Eq. \eqref{event}. 

We introduce the space
\begin{equation}\label{historyspace}
\Xi : = \mathcal{X}^{\times \mathbb{N}}.
\end{equation}
Points in $\Xi$, henceforth called ``histories (of events)'', are denoted by $\underline{\xi} := (\xi_1, \xi_2, \cdots)$, and 
$\Xi$ is called the ``space of histories''.  This space is equipped with the $\sigma$- algebra, $\Sigma$, generated by cylinder sets; (a cylinder set, $\mathcal{C}$, in $\Xi$ is a set of the form  $\mathcal{C}=\Lambda \times \Xi \subset \Xi $, with a base $\Lambda \subseteq \mathcal{X}^{\times n}$). The space of bounded $\Sigma$-measurable functions on $\Xi$ is denoted by $L^{\infty}(\Xi)$. 

%{\color{red} This space is equipped with the $\sigma$- algebra, $\Sigma$, generated by cylinder sets. The algebra has a natural filtration $\Sigma_n \subset \Sigma$, $n \in \mathbb{N}$, generated by cylinder sets, $C$, in $\Xi$ of the form  $C=\Lambda \times \Xi \subset \Xi $, with a base $\Lambda \subseteq \mathcal{X}^{\times n}$. We also use an algebra $\Sigma_{m,n}, n>m,$ which is the largest algebra included in $\Sigma_n$ that has a trivial intersection with $\Sigma_m$.} The space of bounded $\Sigma$-measurable functions on $\Xi$ is denoted by $L^{\infty}(\Xi)$.

We define operators acting on $\mathcal{H}_{S}$ associated with histories of events as follows:
\begin{align}\label{Pi's}
\Pi_{ \underline{\xi}^{(m,n)}  } : =  \Pi_{\xi_m}(t_m) \cdots   \Pi_{\xi_n}(t_n), m<n, \qquad 
\Pi_{ \underline{\xi}^{(n)}  } : =  \Pi_{ \underline{\xi}^{(1,n)}  },
\end{align} 
where  $ \underline{\xi}^{(m,n)} := (\xi_{m}, \xi_{m+1}, \cdots, \xi_{n})  $  labels a stretch of a history 
$\underline{\xi}$ between the $m^{th}$ and the $n^{th}$ event, with $n>m$, and $\underline{\xi}^{(n)}:= \underline{\xi}^{(1,n)}.$ Since the times $t_i, i=1,2,\dots,$ when the instrument $\mathcal{I}$ is triggered do not matter in the following, we will not display them explicitly, anymore; (but we emphasize, once more, that they might depend on the state which the system is prepared in!).

%A physical system, $S$, is characterized, quantum-mechanically, by the following data (see \cite{BlFrSc}): \hspace{1cm} \Pi_{ \underline{\xi}^{(m,n)}  } : =  \Pi_{\xi_m}(t_m) \cdots   \Pi_{\xi_n}(t_n) .  

\begin{assumption}[Decoherence $ \& $ asymptotic abelianness] \label{Decoherence}
The following two properties hold:
\begin{itemize}
\item {\bf Ideal Decoherence:} For every $i \in \{m, \cdots, n \}, 1\leq m < n <\infty$,
\begin{align}\label{deco}
 \hspace{.3cm} \sum_{\xi_i \in \mathcal{X} } &\Pi_{ \underline{\xi}^{(m,n)}  } 
 \big ( \Pi_{ \underline{\xi}^{(m,n)}  } \big )^{*}  = 
 \Pi_{ \underline{\xi}^{(m,i-1)}  }  \Pi_{ \underline{\xi}^{(i+1, n)}  }   
 \big ( \Pi_{ \underline{\xi}^{(m,i-1)}  }  \Pi_{ \underline{\xi}^{(i+1, n)}  } \big )^{*}, 
\end{align}
as an identity between bounded operators acting on $\mathcal{H}_{S}$, with $\Pi_{\underline{\xi}^{(m,k)}} := {\bf{1}}$, for $m>k$.
\item {\bf Asymptotic abelianess:} $ \mathcal{E}_{\infty}$ is contained in the center of $\mathcal{E}^{-}$, (meaning that $\mathcal{E}_{\infty}$ is abelian and operators in $\mathcal{E}_{\infty}$ commute with all operators in $\mathcal{E}^{-}$).
\end{itemize}

\end{assumption}

Next, we note that a state $\omega $ on $\mathcal{E}^{-}$ determines a probability measure 
$\mu_{\omega}$ on $\Xi$: The quantity
\begin{align}\label{mu}
 \mu_{\omega} ( \underline{\xi}^{(m,n)}) : =  \omega \Big (\Pi_{ \underline{\xi}^{(m,n)}  } 
 \big ( \Pi_{ \underline{\xi}^{(m,n)}  } \big )^{*} \Big ) = \sum_{\underline{\xi}^{(m-1)}} \mu_{\omega}(\underline{\xi}^{(n)})
\end{align}
is the probability of measuring $ \xi_i $ at time $t_i$, for every $i = m, m+1, \dots, n$, with $1\leq m < n < \infty$.
%and, if this probability does not vanish,
%\begin{align}
%\frac{ \omega \Big (\Pi_{ \underline{\xi}^{(n)}  } 
%A \big ( \Pi_{ \underline{\xi}^{(n)}  } \big )^{*} \Big ) }{   \mu^{(\omega)} ( \underline{\xi}^{(n)}) },
%\end{align}
% is the expected value of $A$, conditioned to having measured $ \xi_i $ at time $t_i$, for every $i \in \{1, \cdots, n\} $. 
Eq. \eqref{partition of 1} implies that
\begin{equation}\label{partition-of-unity}
\sum_{\xi_{n} \in \mathcal{X}} \Pi_{\underline{\xi}^{n}} \big(\Pi_{\underline{\xi}^{n}}^{*}\big) = \Pi_{\underline{\xi}^{n-1}} \big(\Pi_{\underline{\xi}^{n-1}}^{*}\big) 
\end{equation}
and this guarantees that $\mu_{\omega}$ can be extended to a measure defined on the $\sigma$- algebra 
$\Sigma$; (this follows from the Kolmogorov extension theorem). The measure $\mu_{\omega}$ is well defined whenever $\omega$ is an arbitrary \textit{normal, positive} linear functional on $\mathcal{E}^{-}$, but is not necessarily normalized (i.e., if $\omega$ is a positive multiple of a normal state on $\mathcal{E}^{-}$). 

A measurable set $\Delta \in \Sigma$ is a \textit{``zero-set''} iff $$\mu_{\omega}(\Delta)=0,$$
for an arbitrary normal state $\omega$ on $\mathcal{E}^{-}$. The specification of zero-sets determines what we call a \textit{``measure class''}. The choice of a representation of the algebra $\mathcal{E}$ on a Hilbert space $\mathcal{H}_{S}$ thus determines a measure class. Two measurable functions, $f$ and $g$, on the space $\Xi$ of histories will henceforth be identified iff $$f-g \text{  vanishes almost everywhere},$$
i.e., iff a set $\Delta$ with the property that $f(\underline{\xi})-g(\underline{\xi}) \not= 0, \text{  for any  }\underline{\xi} \in \Delta$ is a zero-set with respect to the measure class determined by the representation of $\mathcal{E}$ on $\mathcal{H}_{S}$.

As already mentioned, the notion of \textit{``tail events''} turns out to be fundamental in our analysis. Here is a precise definition of tail events:  For every $n \in \mathbb{N}$, we denote by $\Sigma_n$ the $\sigma$-algebra on $\Xi$ generated by all sets of the form $\mathcal{X}^{\times n} \times \mathcal{C}$, where $\mathcal{C}$ is a cylinder set in $\Xi$. Furthermore,  we define the \textit{``tail $\sigma$-algebra''}
\begin{align}\label{sigma-infty}
\Sigma_{\infty} := \bigcap_{n\in \mathbb{N}} \Sigma_{n},  
\end{align} 
which, intuitively, consists of all measurable sets, $\Delta$, of histories with the property that if a point 
$\underline{\xi}$ belongs to $\Delta$ and if $\underline{\eta}$ is a history with $\eta_{i}= \xi_{i}$, for all but finitely many $i$, then $\underline{\eta}$ belongs to $\Delta$, too. 
We also use the notation $\Sigma_{m,n}$, $n>m$, for the algebra generated by all the sets $\Delta$ such that $\underline{\xi}, \underline{\eta} \in \Delta $ iff $\xi_{m}=\eta_{m}$, ..., $\xi_{n}=\eta_{n}$.   \\
Complex-valued, bounded functions $f: \Xi \to \mathbb{C}$ measurable with respect to $ \Sigma_{\infty} $ have the property that, almost surely, their values do not depend on any finite number of measurement outcomes in a history; (see Definition \ref{AOI}). Measurable functions on $\Xi$ that are measurable with respect to $\Sigma_{\infty}$ form an algebra, which we henceforth denote by $\mathcal{L}_{\infty}$.

\subsection{New results} \label{NI}

In a previous paper \cite{BaFrFrSc}, we have shown that time-independent properties of a physical system that can be reconstructed from long sequences of repeated direct measurements, using the instrument 
$\mathcal{I}$, can be identified with functions on $\Xi$ that belong to the algebra $\mathcal{L}_{\infty}$ just defined.

A key result proven in the present paper is that, to each pair, $(\omega, f)$, of a normal state $\omega$ on 
$\mathcal{E}^{-}$ and a positive measurable function $f \in \mathcal{L}_{\infty}$, we can associate a new state, denoted by $\omega_{f}$, on $\mathcal{E}^{-}$; see Theorem \ref{Este}, below. 
Our proof involves the construction of a representation of the algebra $\mathcal{L}_{\infty}$ of functions measurable at infinity  in the algebra $\mathcal{E}_{\infty}$ of ``asymptotic observables'' -- see Theorem \ref{rep}. This result enables us to construct a $\Sigma_{\infty}$- ergodic disintegration of the measure 
$\mu_{\omega}$ as a convex combination of mutually singular extremal measures, almost everyone of which corresponds to a normal state on $\mathcal{E}^{-}$; (see Theorem \ref{decom} and Corollary \ref{CMA}). The states that give rise to extremal measures are eigenstates of certain asymptotic observables, (i.e., of certain operators in $\mathcal{E}_{\infty}$). On the support of these extremal measures, functions measurable at infinity are almost surely constant.

%
%.  We prove also several nice mathematical properties of the tail observables: They are obtained by a representation of the commutative $*$-algebra of bounded measurable functions at infinity (measurable with respect to $\mathcal{F}^{(\infty)}$). Moreover, this representation define a PVM (Projections Valued Measure) and, of course, characteristic functions are mapped to projections. States associated to characteristic functions of events are invariant (they remain at infinity) and they represent, therefore, stable properties. Moreover, we prove that the state at infinity  associated to a function $f$ (if we take it as an initial state) gives rise to the measure $  \mu^{\omega^{( \Phi^*(f) )}} $ generated by the density  
%$f$, with respect to the measure $ \mu_{\omega}$.  All these results contribute to clarify an important message we want to convey in this paper:  The facts of the system (or quantities that can be observed in experiments involving many direct measurements) are represented by complex valued (or real valued) bounded functions, measurable with respect to the tail $\sigma$-algebra. The set of these functions is what we called in \cite{BaFrFrSc} the algebra of observables at infinity. In the next lines we give a more precise explanation of the results we just outlined.    
% 
 
\subsubsection{The algebra $\mathcal{L}_{\infty}$ and its image in $\mathcal{E}_{\infty}$} \label{Tail-Obs}

In this subsection we construct a representation, $\Phi \equiv \Phi_{\mathcal{I}}$, of the algebra $\mathcal{L}_{\infty}$ of functions measurable at infinity with values 
%(bounded functions that are measurable with respect to 
%the tail $\sigma-$algebra 
%$ \Sigma_{\infty}$ - see below)
 in the algebra $\mathcal{E}_{\infty}$ of asymptotic observables. The map $\Phi: \mathcal{L}_{\infty} \rightarrow \mathcal{E}_{\infty}$
%we describe how we characterize observables at infinity
%in terms of bounded functions that are measurable with respect to the tail $\sigma-$algebra 
%$ \mathcal{F}^{(\infty)}$ (we also give a precise definition on these concepts), i.e., 
assigns to every function $f \in \mathcal{L}_{\infty}$ an operator  $\Phi(f) \in \mathcal{E}_{\infty}$; see Theorem \ref{rep}, (our main result in this subsection). We will actually define a linear map $\Phi: L^{\infty}(\Xi) \ni f \mapsto \Phi(f) \in \mathcal{E}^{-}$ assigning to every bounded \mbox{$\Sigma$-measurable} function $f$ on $\Xi$ an operator $\Phi(f) \in \mathcal{E}^{-}$. The restriction of this map to 
$\Sigma_{\infty}$-measurable functions will turn out to be the representation of $\mathcal{L}_{\infty}$ in $\mathcal{E}_{\infty}$ we are looking for.  
 % representing an observable at infinity.
%  The existence of the map $\Phi$ is one of the results of this paper. 
%  We also prove several features of this map, that we present below.
   
In this subsection we only describe our results; proofs are deferred to later sections.\\
    
For every $\Sigma$-measurable bounded function $f$, the operator $\Phi(f) \equiv \Phi_{\mathcal{I}}(f) \in \mathcal{E}^{-}$ is defined by the equation 
\begin{align}\label{ok0}
  \omega ( \Phi(f) )  : = \int_{\Xi} f(\underline{\xi}) d \mu_{\omega}(\underline{\xi}) , 
\end{align}
for every normal state $\omega$,  and the definition is extended by linearity to all elements of the pre-dual of $\mathcal{E}^-$. (We recall that all linear functionals on a $C^*$-algebra can be uniquely decomposed into a sum of four positive functionals, see \cite[Theorem~3.3.10]{Mur}). For every bounded $\Sigma$-measurable function $f$ that only depends on finitely many events $\xi_1,\dots, \xi_n$, for some $n<\infty$, of a history 
$\underline{\xi}$, there is an explicit formula for the operator $\Phi(f) $: 
\begin{equation}\label{def of Phi}
\Phi(f):= 
\sum_{\underline{\xi} \in \Xi} f(\underline{\xi})\, \,\Pi_{\underline{\xi}^{(n)}} \big(\Pi_{\underline{\xi}^{(n)}}\big)^{*},
\end{equation}
see formulae \eqref{Pi's} -- \eqref{mu}. Note that if $f\geq 0$ then $\Phi(f)\geq 0$, as an operator, and
$$\Vert \Phi(f) \Vert \leq \Vert f \Vert_{\infty}.$$
Basic properties of the map $\Phi(f)$ are summarized in the following lemma, which will be proven in Section \ref{PHIF}.

\begin{lemma}\label{Phi}
For every $\Sigma-$measurable function $f$, $\Phi(f)$ belongs to $\mathcal{E}^{-}$.  Moreover, the map 
\begin{equation}
\Sigma \ni \Delta \mapsto \Phi(\Delta) \equiv \Phi(\chi_{\Delta})
\end{equation}
is a positive $\sigma$-additive operator-valued measure (POVM).  
\end{lemma}

\begin{definition}[Algebra of asymptotic observables associated with an instrument]\label{AOI}
The algebra, $\mathcal{O}_{\infty}$, of asymptotic observables associated with the instrument $\mathcal{I}$ is the subalgebra of $\mathcal{E}^{-}$ generated by the operators $\Phi(f), f\in \mathcal{L}_{\infty}$, i.e.,
\begin{align}\label{Oinfty}
\mathcal{O}_{\infty} := \langle \Phi_{\mathcal{I}}(f) \vert f \in \mathcal{L}_{\infty} \rangle
\end{align}
\end{definition} 
Using Eq. \eqref{deco} (``Ideal Decoherence''), we see that $\mathcal{O}_{\infty}$ is actually contained in or equal to the algebra $\mathcal{E}_{\infty}$ of asymptotic observables.

 These considerations are summarized in the following theorem, which is one of our main results.

\begin{theorem}[Representation of $\mathcal{L}_{\infty}$ in $\mathcal{E}_{\infty}$]\label{rep}
 The map $$ \mathcal{L}_{\infty} \ni  f \mapsto \Phi(f) \in \mathcal{O}_{\infty} \subseteq \mathcal{E}_{\infty}$$ defines a representation
 (i.e., a homomorphism of $^{*}$algebras) of the algebra $\mathcal{L}_{\infty}$ of functions measurable at infinity with values in the algebra $\mathcal{E}_{\infty}$ of asymptotic observables.
\end{theorem}
%\noindent{\it Proof:}

% for the proofs- Item 1. is proved in \eqref{muybien} and Items 2. and 3. are proved in Section \ref{23}, here we additionally assume asymptotic abelianness, see  \eqref{abelian}). The construction of $\Phi_f$ is  carried in Secion \ref{PHIF} out. 
%\qed

%
%The operator $\Phi_f(A)  $, for every measurable at infinity bounded function $f$ (i.e., measurable with respect to $ \mathcal{F}^{(\infty)} $), represents the  ``Heisenberg'' evolution of the operator $A$ when the time goes to infinity (averaged with respect to the function $f$), as we already pointed out. By Item 1. above, this is characterized by the operator $ \Phi_f(\mathds{1}) $ (or the PVM $\mathcal{F}^{(\infty)} \ni E \mapsto  \Phi_{E}(\mathds{1}) $). Then, the algebra of operators, measurable at infinity, is recovered from a representation of the  
%commutative $*$- algebra of bounded functions (measurable at infinity) in  such a way that characteristic functions of sets are mapped to projections.
 
\subsubsection{States associated with positive elements  in $ \mathcal{L}_\infty $} \label{Tail-States}

In this section we associate with every non-negative function $f\in \mathcal{L}_\infty$ and every state 
$\omega$
%in the algebra of observables at infinity %a state representing 
a bounded, positive linear functional, denoted by $\Phi(f)^{*}(\omega)$, on $\mathcal{E}^{-}$. This functional is given by applying the ``dual map''$, \Phi(f)^* $, to the state $\omega$:
Since $\Phi(f) \geq 0$ if $f \geq 0$, we can take the positive square root, $\Phi(f)^{1/2}$, of the operator $\Phi(f)$ and define a bounded, normal, positive linear functional $\Phi(f)^{*}(\omega)$ on $\mathcal{E}^{-}$ by setting
 \begin{align} \label{dual}
\Phi(f)^*(\omega)(A) : = \omega \big ( \Phi(f)^{1/2} A\Phi(f)^{1/2} \big ), \qquad \text{ for an arbitrary  } A \in \mathcal{E}^{-}.
\end{align}  
If the function $f$ belongs to $\mathcal{L}_{\infty}$ then the operator $\Phi(f)$ belongs to 
$\mathcal{E}_{\infty}$. By \mbox{Assumption \ref{Decoherence}}, (second part: ``Asymptotic abelianess''), every operator in $\mathcal{E}_{\infty}$, and in particular $\Phi(f)$, with $0\leq f \in \mathcal{L}_{\infty}$, and hence $\Phi(f)^{1/2}$, 
belongs to the center of $\mathcal{E}^{-}$, i.e., it commutes with all operators in $\mathcal{E}^{-}$. Thus
\begin{equation*}
\Phi(f)^*(\omega)(A) : = \omega \big ( \Phi(f)^{1/2} A\Phi(f)^{1/2} \big ) = \omega \big(\Phi(f)\, A\big) = 
\omega \big(A\, \Phi(f)\big),
\end{equation*}
for any $A \in \mathcal{E}^{-}$. Moreover, if $0<A=B^{*}\cdot B$ then
\begin{equation}\label{commute}
\Phi(f)^{*}(\omega)(A)=\omega(B^{*}\Phi(f)B).
\end{equation}
This enables us to associate with every non-negative $\Sigma_{\infty}$- measurable function $f$ and every normal state 
$\omega$ on $\mathcal{E}^{-}$ a finite measure $\mu^{f}_{\omega}$ given by
\begin{align}\label{E4}
\mu_{\omega}^f(\Delta) := \int_{\Xi} f(\underline{\xi})\, \chi_{\Delta}(\underline{\xi}) \,
d\mu_{\omega}(\underline{\xi}),  \hspace{1cm} \text{with   }\Delta \in \Sigma. 
\end{align}  
The following theorem (proven in Section \ref{PEste}) is the main result of this section.

\begin{theorem} \label{Este}
For every non-negative function $f \in \mathcal{L}_{\infty}$ and every normal state $\omega$ on $\mathcal{E}^{-}$,
\begin{align}\label{princ1}
\mu_{\omega}^f = \mu_{\Phi(f)^*(\omega) }.
\end{align}
More precisely, every non-negative function $f \in \mathcal{L}_{\infty}$, with $\omega(\Phi(f)) \not= 0$, determines a normal state 
$\omega^{f}:=\omega(\Phi(f))^{-1}\, \Phi(f)^*(\omega)$ on $\mathcal{E}^{-}$ and a finite measure $\mu_{\omega}^{f}$, as in \eqref{E4}, on the space $\Xi$ of histories. 
\end{theorem}

\subsubsection{$ \Sigma_{\infty}$-ergodic disintegration of $\mu_{\omega}$}
         
The measures $\mu_{\omega}$, with $\omega$ a normal state on $\mathcal{E} ^{-}$, encode interesting information on time-independent properties of the system. This is the contents of the next theorem (proven in Section \ref{PMD}).
%Of course, we focus on tail events, that are the ones one expects to be able to identify after many iterative measurements. The goal is to 

\begin{theorem}\label{decom}
 There exists a measure space $  ( \Xi^{\infty},\Sigma^{\infty})$ and a probability measure 
 $P_\omega^\infty$ defined on $\Sigma^{\infty}$ such that        
\begin{align}\label{keypr}
 \mu_{\omega} = \int_{  \Xi^{\infty }}   \mu( \cdot | \,   \nu ) d   P_\omega^\infty ( \nu),
\end{align}
for probability measures $ \Big (   \mu( \cdot | \,   \nu )    \Big )_{ \nu  \in \Xi_\infty } $
defined on $\Sigma$ that are mutually singular and ``$\Sigma_{\infty}-$ergodic'', i.e., 
$\mu(\Delta \vert \nu) = 0 \text{  or  } 1$, for every $\Delta \in \Sigma_{\infty}$. 
\end{theorem} 
Note that the measures $\mu(\cdot \vert \nu)$ depend on our choice of a representation of the algebra $\mathcal{E}$, but not on the choice of a particular normal state $\omega$ on $\mathcal{E}^{-}$.
By the Gel'fand isomorphism, the abelian algebra $\mathcal{O}_{\infty}$ is isomorphic to the algebra of bounded continuous functions on the space $\Xi^{\infty}$. Points in $\Xi^{\infty}$ correspond to spectral values of the operators in $\mathcal{O}_{\infty}$, (``purification'').
Measurable sets $ \Delta \in \Sigma^{\infty}$ represent time-independent physical properties of the system $S$ that can be inferred from very long sequences of direct measurements of $S$, using the instrument 
$\mathcal{I}$, with error rates tending to $0$, as the length of the sequence of measurements of $S$ tends to $\infty$. 

In Section \ref{Choquet}, we prove the following result:
\begin{corollary}\label{CMA}
The measures $ \Big (   \mu( \cdot | \,   \nu )    \Big )_{ \nu  \in \Xi^\infty } $ are extreme points of a certain convex set of measures.
\end{corollary}
Corollary \ref{CMA} implies that \eqref{keypr} is a Choquet-type decomposition. 

In Section \ref{OSO}, we study a simple model with the property that the set $\Xi^{\infty}$ is a countable set of points, so that the integral on the right side of formula \eqref{keypr} is replaced by a sum (or a series), and the points in $\Xi^{\infty}$ are in $1-1$ correspondence with the set of eigenvalues of the time-independent operators in $\mathcal{O}_{\infty}$. This situation has previously been studied in \cite{BaFrFrSc}.

\subsection{Acknowledgments}

\noindent 

M. Ballesteros and M. Fraas are grateful  to J\"urg Fr\"ohlich and G.~M.~Graf   for hospitality at ETH Zurich, where part of this work was carried out. M. Ballesteros is a fellow of the Sistema Nacional de Investigadores (SNI). His research is partially supported by the projects PAPIIT-DGAPA UNAM IN102215 and SEP-CONACYT 254062. B. Schubnel thanks UNAM for hospitality.

\section{Example} \label{SE}

In this section we describe a concrete model system exemplifying our general results: We consider an electron gun shooting electrons into a $T$-shaped conducting channel.  At both ends of the
channel, there are electron detectors, $D_R$ and $D_L$; ($R$ stands for ``right'' and $L$ for ``left''). 
\begin{figure}[H]\label{H}
\begin{center}
\includegraphics[scale=0.30]{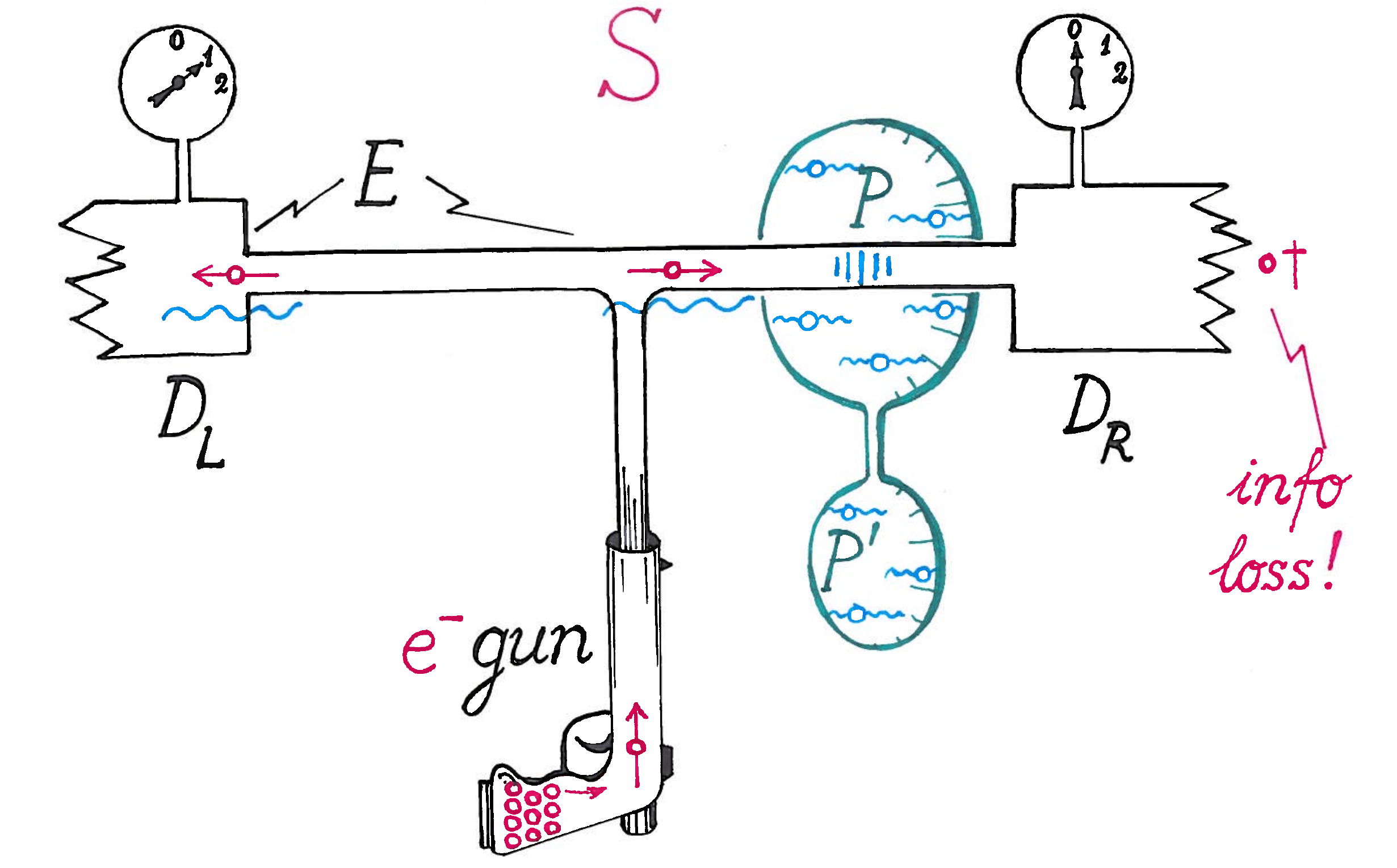}
\caption{The experimental setup.}
\end{center}
\end{figure}
Near the detector $D_R$, underneath the conducting channel, there is a quantum dot, $\overline P$, consisting of two components, $P$ and $P'$, joined by a tunneling junction through which electrons can tunnel from $P$ to $P'$ and back. The dot $\overline{P}$ can bind up to $N< \infty$ electrons that are localized essentially inside $P$ and close to the $T$-channel. These electrons create a \textit{Coulomb blockade} in the right arm of the $T$-channel, while electrons localized inside $P'$ do not have any noticeable effect on the motion of electrons in the $T$-channel. Tunneling of electrons from $P$ to $P'$ happens at a very slow rate, and the operator, $\mathcal{N}_P$, that counts the number of electrons localized inside $P$ is essentially time-independent and is assumed to commute with all operators referring to the electrons in the $T$-channel and in the reservoirs at the ends of the $T$-channel. The significance of $P'$ is that it enables one to prepare electronic states inside the quantum dot $\overline{P}$ that are \textit{not} eigenstates of $\mathcal{N}_P$, but are coherent superpositions of different eigenstates of 
$\mathcal{N}_P$.

The \textit{only} direct measurement or observation possible in the system considered here is to detect whether an electron traveling through the conducting $T$-channel triggers the detector $D_{L}$ at the entrance to the reservoir on the left or the detector $D_{R}$ at the entrance to the reservoir at the right end of the $T$-channel. The symbol ``$-1$'' indicates that the detector $D_{L}$ is triggered, the symbol ``$ 1 $'' stands for a click of $D_{R}$. We interpret these symbols as the eigenvalues of an operator $\hat{X}$ whose spectrum is given by $\lbrace -1,+1\rbrace$, with the property that the eigenvalues $-1$ and $+1$ are both infinitely degenerate. The purpose of sending electrons through the $T$-channel and then observing whether $D_L$ clicks or $D_R$ clicks is to infer from a long sequence of such measurements how many electrons are localized inside $P$, i.e., to determine the eigenvalue of $\mathcal{N}_P$. It has been shown in \cite{BaFrFrSc} that, under suitable hypotheses on the model, every very long sequence of direct measurements of the operator $\hat{X}$ corresponds to a \textit{unique} eigenvalue of $\mathcal{N}_P$, with an error rate that tends to $0$, as the length of the sequence of direct measurements tends to $\infty$. This is the phenomenon of ``purification'' described in \cite{Mass}.

A similar system has actually been realized in the laboratory of Haroche and Raimond and is described in \cite{guerlin}: In their experiment, the quantum dot $\overline{P}$ is replaced by a cavity in which a fairly stable electromagnetic field can be excited and confined; the electrons in the $T$-channel are replaced by Rydberg atoms prepared in a certain fixed coherent superposition of two highly excited internal states, $\vert 1 \rangle$ and $\vert 2 \rangle$. These atoms stream through the cavity, one by one, without emitting or absorbing any photons. During its sojourn inside the cavity, the internal state of each atom precesses in the space spanned by the states $\vert 1 \rangle$ and $\vert 2 \rangle$ with an angular velocity that depends on the number, $\nu$, of photons in the cavity. The detectors $D_L$ and $D_R$ are replaced by a device measuring a projection of the internal state of every Rydberg atom after it has emerged from the cavity. The two possible outcomes of this measurement can be represented by the symbols 
$\pm 1$, as in the system considered above.

 In both systems, a history, $\underline{\xi}$, is a sequence of symbols, $\xi_{i} = \pm 1, \, i=1,2,3,...,$ corresponding to the outcomes of infinitely many direct measurements of the observable represented by the operator $\hat{X}$ ($=$ clicks of one of the two detectors, in the first model), as described above. The first part of Assumption \ref{Decoherence} is very nearly satisfied, as subsequent electrons moving through the $T$-channel (or Rydberg atoms streaming through the cavity) are essentially \textit{independent} of one another. Furthermore, under certain hypotheses on the dynamics of the system, the algebra 
 $\mathcal{E}_{\infty}$ turns out to be isomorphic to the algebra generated by all bounded functions of the operator $\mathcal{N}_P$, (or of the photon number operator, respectively), and satisfies the second part of Assumption  \ref{Decoherence}. Thus, the algebra $\mathcal{E}_{\infty}$ is isomorphic to the algebra of bounded functions on the set 
 $\lbrace 0, \dots ,N \rbrace$ of eigenvalues of $\mathcal{N}_{P}$, which, under natural assumptions, is generated by the frequency, 
 $p_{1}(\underline{\xi}) \in \mathcal{L}_{\infty}$, of the symbol $1$ in a history $\underline{\xi}$. (As shown in \cite{BaFrFrSc}, almost every history $\underline{\xi} \in \Xi$ corresponds to a unique eigenvalue 
 $\nu = \nu(\underline{\xi}) \in \lbrace 0,1,\dots, N \rbrace$ of the operator
 $\mathcal{N}_{P}$, and $\nu(\underline{\xi})$ corresponds to a unique value of the frequency 
  $p_{1}(\underline{\xi})$ of the symbol $1$ in the history $\underline{\xi}$.)

 %A cylinder set in $\Xi$  is a set of the form 
%\begin{align}
%\Sigma[ \eta_1, \eta_2, \cdots, \eta_n] := \{  \underline \xi \in \Xi | \xi_i = \eta _i, i \in \{1,2, \cdots, n \} \}. 
%\end{align}
%The $\sigma$-algebra of events ($\mathcal{F}$) is the $\sigma$-algebra generated by the cylinder sets. For every $n \in \mathbb{N}$, %denote by 
%$\mathcal{F}^{(n)} $ is the sigma algebra
%in $ \Xi $ generated all sets of the form $ \sigma^{n} \times \Sigma[ \eta_1, \eta_2, \cdots, \eta_m] $ (here $n$ is fixed and $m$ varies).  The tail $\sigma$-algebra (representing events at infinity) is 
%\begin{align}
%\mathcal{F}^{(\infty)} := \bigcap_{n\in \mathbb{N}}\mathcal{F}^{(n)}.   
%\end{align} 
We denote by $\Pi_{\xi}$ the spectral projection of the operator $\hat{X}$ corresponding to the eigenvalue 
$\xi=\pm 1$, and by $Q_{\nu}$ the spectral projection of the operator $\mathcal{N}_{P}$ corresponding to the eigenvalue \mbox{$\nu \in \lbrace 0, \dots, N \rbrace$}. The only instrument $\mathcal{I}$ in our model is generated by the two projections 
$\lbrace \Pi_{-1}, \Pi_{+1}\rbrace$. By $\Pi_{\xi}(t )$ we denote the orthogonal projection representing 
$\Pi_{\xi}$ in the algebra $\mathcal{E}_{\geq t}$; see item IV, Sect. \ref{GF}. Whenever an electron travels through the $T$-channel and reaches a reservoir at the left or right end of the $T$-channel the instrument $\mathcal{I}$ is triggered and exhibits the value $-1$ if $D_L$ has clicked and the value $+1$ if $D_R$ has clicked. 
Let $t_i$ denote the time at which the $i^{th}$ direct measurement of $\hat{X}$ has been completed, corresponding to a click of one of the two detectors $D_L, D_R$, and let $\xi_{i}$ denote the result of the $i^{th}$ measurement.
Then the probability of observing the sequence $\lbrace \xi_{1}, \dots,\xi_{n} \rbrace$ in the first $n$ direct measurements is given by
\begin{align}\label{pto1}
\mu_{\omega}( \xi_1, \cdots, \xi_n  ) := \omega\Big(  \pi_{\xi_1} \cdots \pi_{\xi_n}   \pi_{\xi_n} \cdots \pi_{\xi_1}  \Big),
\end{align}    
where $\omega$ is the state of the total system, and $\pi_{\xi_i} := \Pi_{\xi_i}(t_i)$; (this is the so-called L\"{u}ders-Schwinger-Wigner- or LSW-formula; see \eqref{mu}.)
% 
% 
% Suppose that the initial state $\omega$ is actually a density matrix $\rho$ and define 
%\begin{align}
%\rho^{(n)}(  \xi_1, \cdots, \xi_n   ) : = \frac{ \pi_{\xi_n} \cdots \pi_{\xi_1} \rho  \pi_{\xi_1} \cdots \pi_{\xi_n}}{\mu_{\omega}( \xi_1, \cdots, \xi_n  ) }, 
%\end{align}  
%representing the density matrix of the system conditioned to having measured $(  \xi_1, \cdots, \xi_n   )$.
%In \cite{BaFrFrSc} we prove that the sequence $ ( \rho^{(n)} )_{n \in \mathbb{N} } $ has a limit 
%$ \rho^{(\infty)} : \Xi \to \text{(Density Matrices)} $. $\rho^{(\infty)}$ is the state of the system after infinitely many measurements occurred. Evidently, it depends on the histories (it is, therefore, a measurable function with domain $\Xi$ and values in the set of density matrices).    
In \cite{BaFrFrSc} we have shown that, under certain assumptions on the model, the Choquet decomposition of the measure 
$\mu_{\omega}$ on the space $\Xi$ of histories determined by \eqref{pto1}, as described in
Theorem \ref{decom}, reduces to a finite sum:
\begin{align}\label{ok1}
\mu_{\omega} = \sum_{\nu = 0}^{N} p_\nu(\omega)\, \mu( \cdot  |\nu ),
\end{align}
where $$p_{\nu}(\omega)=\omega(Q_{\nu})\geq 0$$
 is the Born probability of observing $\nu$ electrons bound to $P$.
The space $\Xi$ of histories has a decomposition into mutually disjoint  subsets 
$\Xi_\nu \in \Sigma_{\infty} $, $\nu \in \{ 0, \cdots, N\} $, with the property that the union of these subsets has full measure (w. r. to any of the measures $\mu_{\omega}$) and that, for every set $\Delta \in \Sigma$,
\begin{align}
\mu(\Delta | \nu)  = \frac{1}{p_{\nu}(\omega)}\mu_{\omega}(\Delta \cap  \Xi_\nu).
\end{align}
The measures $\mu(\cdot | \nu)$ satisfy a $0$-$1$ law when restricted to $\Sigma_{\infty}$. A set of histories contained in the subset $\Xi_{\nu}$ corresponds to selecting an eigenstate of $\mathcal{N}_{P}$ corresponding to the eigenvalue $\nu$, for $\nu = 0,1,\dots,N$. 
%It follows that 
%\begin{align}
%\mathcal{O}^{(\infty)} \cong L^{(\infty)}( \{ 1, \cdots, N \}),  
%\end{align}
%actually the characteristic functions  $ \{ \chi_{\Xi_\nu} \}_{\nu = 1}^{N} $ form a basis of 
%   $ \mathcal{O}^{(\infty)} $. 

Let us add a little further precision to this discussion. We assume that subsequent direct measurements of the observable $\hat{X}$ are strictly \textit{independent} of each other. This can be expressed as the property that the measures 
$\mu_{\omega}$ are \textit{exchangeable}, i.e.,
$$\mu_{\omega}(\xi_{\sigma(1)}, \dots \xi_{\sigma(n)}) \text{  is } independent \text{  of the permutation  }\sigma, \quad \forall \text{  permutations  } \sigma,  \forall  n=1,2,3, \dots$$
Then \textit{de Finetti's theorem} says that the measures $\mu(\cdot \vert \nu)$ are product measures, i.e.,
\begin{equation}
\mu(\xi_1, \dots, \xi_n \vert \nu) = \prod_{i=1}^{n} p(\xi_i \vert \nu), \qquad \forall n=1,2,3, \dots,
\end{equation}
for some probabilities $p(\xi \vert \nu)$ with $\sum_{\xi} p(\xi \vert \nu) = 1$, for all $\nu = 0,1,\dots,N$.
The interpretation of the probabilities $p(\xi\vert \nu)$ is that $p(1\vert \nu)$ is the frequency of the measurement outcome $\xi = 1$ in every history $\underline{\xi}$ belonging to the subset $\Xi_{\nu}$, for $\nu = 0,1,\dots,N$. Henceforth it is assumed that the function $p(1\vert \nu)$ separates points in $\lbrace 0,\dots , N \rbrace$; (see \cite{BaFrFrSc}).

Let  $f$ be a bounded
   $\Sigma_{\infty}-$measurable function. Then $$f=\sum_{\nu=0}^{N} f_{\nu} \chi_{\Xi_{\nu}},$$ 
   with $f_{\nu} \in \mathbb{C}$, and $\Phi(f)=\sum_{\nu=0}^{N}f_{\nu}Q_{\nu}$. If $f\geq 0$ then
   \begin{equation}\label{state-dec}
   \Phi(f)^{*}(\omega) = \sum_{\nu=0}^{N} f_{\nu} \, \omega_{\nu},
   \end{equation} 
  where
  $$\omega_{\nu}(A):=\omega(Q_{\nu}AQ_{\nu}), \qquad \text{for an arbitrary    }A \in B(\mathcal{H}_{S}).$$ 
  
To prove \eqref{state-dec}, one shows that the operators $\Phi(f)$ and $\Phi(f)^{1/2}$, with f an arbitrary bounded non-negative $\Sigma_{\infty}-$measurable function on $\Xi$,
are functions of $\mathcal{N}_{P}$. For details see \cite{BaFrFrSc}.

\section{Proofs}  \label{Proofs}

{
We start with a key technical lemma that we use to approximate $\Phi(f)$ by operators of the form Eq.~(\ref{def of Phi}). 
\begin{lemma}
\label{lem:aux}
For every function $f \in L^\infty(\Xi)$ there exists an increasing sequence $j_n$ and a sequence of functions $f_n \in L^\infty(\Xi)$ measurable with respect to the algebra $\Sigma_{j_n}$ such that
$$
w-\lim_{n \to \infty} \Phi(|f_n - f|) = 0.
$$
If in addition $f$ is measurable with respect to the algebra $\Sigma_\infty$ then $f_n$ can be chosen so that it is measurable with respect to the algebra $\Sigma_{i_n,j_n}$ for some increasing sequences $(i_n)_{n \in \mathbb{N}},\,(j_n)_{n \in \mathbb{N}}$ with $i_n<j_n$.
\end{lemma}

 \begin{proof}
The proof relies on standard arguments in measure theory. We only  prove the statement of the lemma  in the case where $f$ is $\Sigma_\infty$-measurable. Most of the arguments  used in our proof apply to the general case without any modification. We first focus on characteristic functions. The following assertion is not difficult to prove (see for instance \cite{Bill}, Theorem 11.4, for similar results): \\

\noindent \textbf{(Approximation by cylinder sets)}  Let  $\Delta \in \Sigma$ and let $(\omega^{(m)})_{m \in \mathbb{N}}$ be a sequence of normal states in $\mathcal{E}$.  Then there exists a sequence $ (\Delta_n)_{n \in \mathbb{N}}$ of cylinder sets such that
 \begin{equation}
 \mu_{\omega^{(m)}}(\Delta_n \setminus \Delta)+ \mu_{\omega^{(m)}}( \Delta  \setminus \Delta_n) <\frac{1}{n}
 \end{equation}
 for all  $m  \leq n$. If $\Delta \in \Sigma_{\infty}$, then the sets $\Delta_n $ can be chosen to be in 
 $\Sigma_{i_n,j_n}$ for some increasing sequences $(i_n)_{n \in \mathbb{N}},\,(j_n)_{n \in \mathbb{N}}$ with $i_n<j_n$.\\

Let $\Delta \in \Sigma_{\infty}$. The assertion above implies the existence of a sequence of cylinder sets $(\Delta_n )_{n \in \mathbb{N}}$, such that  $\chi_{\Delta_{n}}$ converges to $\chi_{\Delta}$ in $L^1(\Xi, \mu_{\omega^{(m)}}, \Sigma)$ for all $m \in \mathbb{N}$. Moreover, $\Delta_n \in  
 \Sigma_{i_n,j_n}$ for some increasing sequences $(i_n)_{n \in \mathbb{N}},\,(j_n)_{n \in \mathbb{N}}$ with $i_n<j_n$.  By definition of the map $\Phi$ (see Eq. \eqref{ok0}), we have that 
\begin{align} \label{weak11}
\omega^{(m)}(\Phi( \vert \chi_{\Delta_{n}} - \chi_{\Delta} \vert ))\underset{n \rightarrow \infty}{\longrightarrow} 0
\end{align}
for all $m \in \mathbb{N}$. The Hilbert space $\mathcal{H}$ being separable, we deduce from  \eqref{weak11} that $\omega(\Phi( \vert \chi_{\Delta_{n}} -\chi_{\Delta} \vert ))$ converges to zero for every  state $\omega$ in $\mathcal{E}^{-}$ by a density argument. It follows from Jordan's decomposition  theorem (see e.g. \cite[Theorem, 3.3.10]{Mur}) that $\Phi(\chi_{\Delta_{n}})$ converges weakly to $\Phi(\chi_{\Delta})$. In particular $\Phi(\chi_{\Delta})$ belongs to $\mathcal{E}^{-}$ (the von-Neumann algebra $\mathcal{E}^{-}$ is weakly closed) and $\Phi(\chi_{\Delta})$ is bounded in norm because $|\omega(\Phi(\chi_{\Delta}))| \leq 4 \|  \omega \|  $ for every bounded linear functional $\omega$ on $\mathcal{E}^{-}$.  The proof of the  same results for an arbitrary function $f \in L^\infty(\Xi)$ measurable with respect to the $\sigma$-algebra $\Sigma_\infty$ follows directly from a density argument using simple functions.
 \end{proof}
 
\subsection{Proof of Lemma \ref{Phi}}\label{PHIF}
Let $f \in  L^\infty(\Xi)$. For every positive linear functional $\omega$, we have that  
$$
\left| \int_\Xi f(\underline{\xi}) \mathrm{d} \mu_\omega(\underline{\xi}) \right| \leq \int_\Xi ||f||_\infty \mathrm{d} \mu_\omega(\underline{\xi}) = ||f_\infty|| \omega(\id).
$$
This bound and Lemma \ref{lem:aux} imply that $\Phi(f) \in \mathcal{E}^-$. The properties of  a POVM are inherited from the corresponding properties of integrals in Eq.~(\ref{ok0}). For example $\omega(\Phi(\Xi)) = 1$ for all states $\omega$ implies $\Phi(\Xi) = \id$.  Next we remind that, by definition,
\begin{equation}
\omega\big ( \Phi(\Delta)\big ) =   \mu_{ \omega }(\Delta), 
\end{equation}
for every set $\Delta \in \Sigma$ and for any positive functional $\omega$. Therefore, for any disjoint sequence $(\Delta_{n})_{n \in \mathbb{N}} \in \Sigma$,
\begin{equation}
 \omega\Big ( \Phi\big (\bigcup_n \Delta_n\big )\Big )  =  
 \mu_{ \omega }\Big ( \bigcup_n \Delta_n\Big ) = 
 \sum_n  \mu_{ \omega }\big (  \Delta_n \big ) = \sum_n   \omega \big (  \Phi (\Delta_n ) \big ). 
\end{equation}
 This shows that \begin{equation}\label{canijo}
\Phi (\bigcup_n \Delta_n ) = \sum_n \Phi (\Delta_n ),
\end{equation}
where the series converges in the $\sigma$-weak topology (i.e. it also converges weakly).  \qed 

\subsection{Proof of Theorem \ref{rep}} \label{Prep}
Let $f \in \mathcal{L}_\infty$. Consider a sequence of functions $(f_n)_{n \in \mathbb{N}}$ as in  Lemma~\ref{lem:aux},  such that $\Phi(f_n) \to \Phi(f)$ weakly  and that $f_n$ is measurable with respect to the algebra $\Sigma_{i_n,j_n}$
 for some strictly increasing sequences $ i_n, j_n$  with $i_n < j_n$. Using Assumption \ref{Decoherence}, we have that
\begin{align}\label{Phifinfty}
\Phi(f_n)  =   \sum_{\underline{\xi}^{(i_n, j_n)}} 
f_n \Big ( \underline{\xi}^{(i_n, j_n)} \Big )\Pi_{ \underline{\xi}^{(i_n, j_n)}  } 
  \big ( \Pi_{  \underline{\xi}^{(i_n, j_n)}  } \big )^{*} ,
\end{align}
and hence $\Phi(f_m)$ belongs to $ \mathcal{E}_{\geq t_{i_n}} $, for every 
  $m 
  \geq n$, (see Eq. \eqref{Pi's}).  We conclude that 
 $\Phi(f) \in \mathcal{E}_{\infty} $.

Next, we show that $\Phi\vert_{\mathcal{L}_{\infty}}$ is a $^{*}$homomorphism from $\mathcal{L}_{\infty}$ to $\mathcal{E}_{\infty}$. It suffices to show that $ \Phi(f\cdot \chi_{\Delta}) =  
 \Phi(f) \cdot \Phi(\chi_{\Delta}) $ for any cylinder set $\Delta \in \Sigma$.  The previous equality then easily extends to arbitrary functions $f, g  \in \mathcal{L}_\infty$ by a density argument. (Moreover, compatibility with the $^{*}$ operation is obvious). 
 Let $f_n$ be approximations of the function $f$ as in the first part of the proof and assume that $\Delta  \in \Sigma_{m}$ for some $m \in \mathbb{N}$. If $n$ is such that $i_n>m$, we have that 
 \begin{align} \label{deco1}
 \Phi(f_n\cdot \chi_{\Delta})&=\sum_{\underline{\xi}^{(m)},\underline{\xi}^{(i_n,j_n)} } \chi_{\Delta} (\underline{\xi}^{(m)})  f_n(\underline{\xi}^{(i_n,j_n)}) \text{ } 
 \Pi_{\xi_1} \, ...  \, \Pi_{\xi_m}  \Pi_{\xi_{i_n}} \,...\, \Pi_{\xi_{j_n} } \Pi_{\xi_{j_n}} \,... \, \Pi_{\xi_{i_n}} \Pi_{\xi_m} \, ... \, \Pi_{\xi_1}
 \end{align}
 as a consequence of the decoherence assumption in Assumption \ref{Decoherence}. Taking the weak limit on both sides of Eq. \ref{deco1}, we deduce that
 \begin{align*} 
 \Phi(f\cdot \chi_{\Delta})&=\sum_{\underline{\xi}^{(m)}} \chi_{\Delta} (\underline{\xi}^{(m)})  \Pi_{\xi_1} \, ...  \, \Pi_{\xi_m}  \Phi(f) \Pi_{\xi_m} \, ... \, \Pi_{\xi_1}.
 \end{align*}
Using now that  $\mathcal{E}_{\infty}$ is contained in the center of $\mathcal{E}^{-}$ (Asymptotic abelianess, see Assumption \ref{Decoherence}) and that  $\Phi(f) \in \mathcal{E}_{\infty}$, we can commute $\Phi(f)$ with the projections $\Pi_{\xi_k}$ for any $k \in \{1,...,m\}$.  Therefore,
\begin{equation} \label{comu}
 \Phi(f \chi_{\Delta}) = \Phi(f) \Phi(\chi_{\Delta}).
\end{equation}
\qed 

\subsection{Proof of Theorem \ref{Este}}\label{PEste}
Let $f \in \mathcal{L}_\infty$. For a characteristic function $\chi_\Delta$ of a cylinder set $\Delta$, we have, using asymptotic abelianess (Assumption~\ref{Decoherence}), that
$$
\mu_{\Phi(f)^*\omega}(\Delta) = \omega(\Phi(f) \Phi(\chi_\Delta)).
$$
In the proof of Theorem~\ref{rep}, Eq.~(\ref{comu}), we have established that $\Phi(f) \Phi(\chi_\Delta) = \Phi(f \chi_\Delta)$. Hence $\mu_{\Phi(f)^*\omega}(\Delta) = \mu_\omega^f(\Delta).$
  
  \qed
}
\subsection{Proof of Theorem \ref{decom}} \label{PMD}

 We follow the formalism in \cite{SchwartzTata} (which is a review of \cite{Schwartz}), where the reader finds most of the results relevant for our purposes. Similar results can also be found in \cite{Shi}, \cite{Rau} and Theorem 2.111 in \cite{S1995}.

Formula \eqref{keypr} represents a special case of 
$ \Sigma_{\infty} $-ergodic disintegration of measures -- in this paper applied to the measures 
$\mu_{\omega}$. Specifically, given a measure $ \mu_{\omega} $ on $(\Xi, \Sigma)$, where $\omega$ is a normal state on 
$\mathcal{E}^{-}$, the $\Sigma_{\infty}$-ergodic disintegration of $\mu_{\omega}$ consists of a  family of probability measures $\Big ( \mu( \cdot | \underline \xi) \Big )_{\underline \xi \in \Xi }$ with the properties that
\begin{enumerate}
\item for every $\Delta  \in \Sigma $,  the function $  \Xi \ni \underline \xi  \to \mu(\Delta | \underline \xi )  $
is measurable with respect to $\Sigma_{\infty}$;  
\item  for every $\Delta  \in \Sigma $ and every $\Lambda \in \Sigma_{\infty}$
\begin{align}\label{rcp}
\int_\Lambda \mu(\Delta | \underline \xi ) d  \mu_{\omega}(\underline \xi) =    \mu_{\omega}(\Lambda  \cap \Delta); 
\end{align}
\item for almost every $ \underline \xi \in \Xi   $  (with respect to the measures $\mu_{\omega}$) 
$\mu( \cdot | \underline \xi)$ is $\Sigma_\infty -$ergodic, i.e.,  
\begin{align}  \label{hay}
   \mu(\Lambda | \underline \xi) \in \{ 0, 1  \},    
\end{align} 
for every $\Lambda \in \Sigma_{\infty}$.
\end{enumerate}

% Juerg's comment

%{\color{red} I think it is important to observe that the measures $\mu(\cdot \vert \underline{\xi})$ do actually \textit{NOT} depend on the choice of the state $\omega$!}

Proposition 103 in \cite{SchwartzTata} implies that  the measures $\mu_{\omega} $, with $\omega$ a normal state on $\mathcal{E}^{-}$, have a unique $\Sigma_{\infty}$-ergodic decomposition, as described above. The existence of such a decomposition is not entirely obvious; it is a property that is stronger than the existence of conditional expectations, as the latter property does not require ergodicity, (Item 3, above). Item 3 is a somewhat delicate issue and is based on certain specific hypotheses. 

Next, we briefly explain why Proposition 103 in \cite{SchwartzTata} is applicable in our situation. The hypotheses on $\Sigma$ are listed on page 113 in \cite{SchwartzTata}.
Since $\mathcal{X}$ is a compact Polish space, $ \Xi $  is a compact Polish space, too; see \eqref{historyspace}.  
Hence the sigma algebra $\Sigma$ is countably generated, and it satisfies properties (1)-(4) on page 113 in \cite{SchwartzTata} .
A $\sigma$-algebra $\Omega \subset \Sigma$ for which an $\Omega$-ergodic decomposition (satisfying Items 1.-3. above, with $\Omega$ replacing  $\Sigma_{\infty} $) of the measures $\mu_{\omega}$ exists  is called, in the terminology of \cite{SchwartzTata},   ``$\mu_{\omega}$-strong''; (see Proposition 102 in \cite{SchwartzTata} ). All $\sigma$-algebras $\Sigma_{n}$, for every $n$, are $\mu_{\omega}-$strong, (as they are countably generated -- see Proposition 99 in \cite{SchwartzTata}  ). Because the intersection of 
$\mu_{\omega}$-strong $\sigma$-algebras is  $\mu_{\omega}$-strong, (see Proposition 103 in \cite{SchwartzTata} ), it follows that $\Sigma_{\infty}$ is $ \mu_{\omega}$-strong, too. Thus, every measure 
$\mu_{\omega}$ (with $\omega$ a normal state on $\mathcal{E}^{-}$) has a $\Sigma_{\infty}$-ergodic decomposition the uniqueness of which follows from the uniqueness of conditional expectations; see Theorem 36 in \cite{SchwartzTata}. (Notice that the measures $\mu(\cdot | \underline \xi) $ are Radon measures -- see proof of Theorem 38 in \cite{SchwartzTata}.)  

Eq. \eqref{rcp} implies, taking $\Lambda = \Xi$, that 
\begin{equation}\label{rcp1}
  \mu_{\omega}(\cdot) = \int \mu( \cdot | \underline \xi) d  \mu_{\omega}(\underline \xi) ,
\end{equation}
which is close to formula \eqref{keypr}, but is weaker than \eqref{keypr}: While the measures 
$\mu( \cdot | \underline \xi) $ $(\underline \xi \in \Xi)$  in Eq. \eqref{rcp1} are 
$\Sigma_{\infty}-$ergodic, they are not necessarily mutually singular. Following \cite{SchwartzTata},  Chapter 7, Section 3, we can construct a $\Sigma_{\infty}$-ergodic decomposition with \textit{mutually singular} measures. Of course, we cannot parametrize the $\Sigma_{\infty}-$ergodic mutually singular measures  
in a $\Sigma_{\infty}$-ergodic decomposition of $\mu_{\omega}$ by the points (histories) $\underline \xi \in \Xi$, and we need to come up with an appropriate parametrization.

We say that a measure $\rho$ is disintegrated with respect to 
$ \Sigma_{\infty} $ by $  \Big ( 
\mu( \cdot | \underline \xi) \Big )_{\underline \xi \in \Xi} $ if, for every $\Delta  \in \Sigma $ and every $\Lambda \in \Sigma_{\infty}$,
\begin{align}\label{rcpto}
\int_{\Lambda} \mu(\Delta | \underline \xi ) d  \rho (\underline \xi) =    \rho (\Lambda \cap \Delta).  
\end{align}
For every $ \underline \xi \in \Xi$, we define 
\begin{align} 
 {\rm Mol }(\underline \xi) : = \Big \{ \underline \xi' \in \Xi \: \Big | \: \mu( \cdot | \underline \xi) = 
 \mu( \cdot | \underline \xi')   \Big \}.    
\end{align}
The set $   {\rm Mol }(\underline \xi)  $ belongs to $\Sigma_{\infty} $. To see this, we choose a family 
$\big( \Delta_n \big)_{n \in \mathbb{N}}$ of $\Sigma$-measurable sets generating the $\sigma$-algebra 
$\Sigma$.
Since the functions $ \underline{\xi} \mapsto  \mu(\Delta_n | \underline \xi )$ are $\Sigma_{\infty}$-measurable, the set \mbox{$\Big \{ \underline \xi' \in \Xi \: \Big | \: \mu(\Delta_n | \underline \xi ) =
 \mu(\Delta_n | \underline \xi' )   \Big \} $}  belongs to $\Sigma_{\infty}$, for all $n$. The intersection over $n$ of all these sets, which equals $  {\rm Mol }(\underline \xi)  $, then belongs to  $\Sigma_{\infty}$, too.  The symbol $  {\rm Mol }  $ stands for ``molecule''; (notice that  $  {\rm Mol }(\underline \xi)  $ is a union of  ``atoms'').
We define (see Theorem 107 and Proposition 105 in \cite{SchwartzTata}) a set
$\widetilde{\Xi}$ by 
\begin{align}
\widetilde{\Xi} : = \Big \{ \underline{\xi } \in \Xi  \: \Big | \: 
\mu( \cdot | \underline \xi ) \text{  is disintegrated with respect to $ \Sigma_{\infty} $ by $  \Big ( 
\mu( \cdot | \underline \xi') \Big )_{\underline \xi' \in \Xi} \,\,\& \,\,
\mu( {\rm Mol }\big( \underline \xi \big) |  \underline \xi) =  1 $} \Big \}.
\end{align}
In Proposition 105 in \cite{SchwartzTata}
it is proven that $\widetilde \Xi  \in \Sigma$ and $\mu_{\omega}( \widetilde \Xi ) = 1 $. The set 
$ \widetilde \Xi $ serves to avoid repetitions in the decomposition \eqref{rcp1}: We propose to identify all points in a given  molecule ${\rm Mol }\big( \underline \xi \big)$. Let
$ \Xi^\infty  $ be the quotient set of $ \widetilde{\Xi} $ modulo molecules, identifying every molecule with a point. We let $p$ be the natural projection from ${\widetilde{\Xi}}$ onto $\Xi^\infty$, and we define the sigma algebra 
$\Sigma^\infty $ to be generated by the sets $ E \subset \Xi^\infty $ with the property that 
$p^{-1}(E) \in \Sigma_{\infty}$. We then define a measure $P_{\omega}^\infty$ on 
$(\Xi^{\infty}, \Sigma^{\infty})$ by setting 
$$P_{\omega}^\infty (E) =   \mu_{\omega}(p^{-1}(E)), \text{   for every   } E \in \Sigma^{\infty}.$$
 Furthermore, for an arbitrary 
$\nu \in \Xi^{\infty}$ with $\nu=p(\underline{\xi}), \,\underline{\xi} \in \widetilde{\Xi}$, we define
$$\mu( \cdot | \nu) :=  \mu( \cdot |  \underline{\xi} ),$$
 with $\mu(\cdot \vert \underline{\xi})$ as above. Since the molecules form a partition of  $\widetilde{\Xi}$ by $\Sigma_{\infty}$-measurable sets, and, for every $\underline{\xi} \in \widetilde{\Xi}$, 
 $\mu({\rm Mol}( \underline{\xi} )  |\underline{\xi}) =1$, it follows that the measures  
 $\Big(\mu( \cdot  | \nu )   \Big)_{ \nu \in \Xi^\infty  }$ are mutually singular and $ \Sigma_{\infty}-$ergodic. Furthermore, for every measurable set $E \in \Sigma^\infty $ and every $\Delta \in \Sigma $
\begin{align}
 \mu_{\omega}(\Delta \cap p^{-1}(E)) = \int_{  p^{-1}(E) }   \mu(\Delta| \underline{\xi} )\, d \mu_{\omega}( \underline{\xi} ) = \int_{ E }   
 \mu(\Delta|  \nu )\, d  P_{\omega}^\infty ( \nu ). 
\end{align} 
In particular, 
\begin{align}\label{key}
 \mu_{\omega}(\cdot) = \int   \mu( \cdot | \nu ) \, d  P_{\omega}^\infty ( \nu ). 
\end{align}
This is the ergodic decomposition announced in \eqref{keypr}. 
\qed

%  of $\mu^{(\omega)}  $ that generalizes \eqref{ok}, in terms of $ \mathcal{F}^{(\infty)}-$ergodic and mutually singular measures. The elements $ \dot{\underline{\xi}} \in \dot{\widetilde{\Xi}}   $ represent facts of the system that, in principle,  can be observed after an infinite number of direct measurements. 
%However, if the initial state is $\omega$, then not all the the quantities  
%$  \dot{\underline{\xi}} \in \dot{\widetilde{\Xi}} $ are accessible at infinity, because some of them have zero probability to occur (this happens only because we fixed the initial state).     
%Then (given an initial state $\omega$), the facts of the system (accessible at infinity)  is given by $ p({\rm supp}\big ( \mu^{(\omega)}) \big ) $, where ${\rm supp}  $ stands for support. In the case the that  $ p({\rm supp}\big ( \mu^{(\omega)}) \big ) $ is finite, we arrive at Eq. \eqref{ok}, where the elements of  $ p({\rm supp}\big ( \mu^{(\omega)}) \big )   $
%are denoted by $\nu \in  \{1, \cdots, N \}$.   
 
\subsection{Proof of Corollary \ref{CMA}}\label{Choquet}

We recall that Choquet theory is designed to represent elements in certain convex spaces as unique convex combinations (more generally, integrals given in terms of probability measures) of extremal points of those spaces. Here we explain how Eq. \eqref{key} can be viewed as an integral, given in terms of the probability measure $dP^{\infty}_{\omega}$, of extremal probability measures, 
$\mu(\cdot \vert \nu)$, belonging to a convex space of measures.
We define the set, $\mathcal{K}$, of probability measures by
\begin{align}
\mathcal{K}  :=  & \Big \{ \mu \Big | \text{ $ \mu $ is  a Radon probability measure on $(\Xi, \Sigma)$} \\ \notag  & \hspace{0.7cm} \text{  disintegrated with respect to $ \Sigma^{\infty} $ by 
$\Big( \mu( \cdot | \nu) \Big )_{\nu \in \Xi^{\infty}} $ }  \Big \}. 
\end{align}
The set $\mathcal{K}$ is obviously convex. Moreover,
in Proposition 105 of \cite{SchwartzTata} it is proven that $\mu(\widetilde{\Xi})=$ \mbox{$ \mu(p^{-1}(\Xi^{\infty}) ) = 1$,} for every 
$\mu \in \mathcal{K}$.  Finally, Theorem 106 in \cite{SchwartzTata} asserts that the measures 
$ \Big( \mu( \cdot | \nu) \Big)_{\nu \in \Xi^{\infty}}$ are the extremal points in 
$\mathcal{K}$. We conclude that Eq. \eqref{key} is the Choquet decomposition of a measure 
$\mu_{\omega} \in \mathcal{K}$. 
%
%
%Thus, the elements $\dot{\underline{\xi}} \in \dot{\widetilde{\Xi}}  $ represent pure quantities of the system,   since the corresponding measures $   \mu^{(\omega)}_{\dot{\underline{\xi}}} $ are  extremal points. Then, an interpretation of Eq. \eqref{key} might be  that the system purifies at infinity, because the measures  $   \mu^{(\omega)}_{\dot{\underline{\xi}}} $ are mutually singular. 

\subsection{On the spectrum of $\mathcal{O}_{\infty}=\Phi_{\mathcal{I}}(\mathcal{L}_{\infty})$}\label{OSO}

In this subsection we study the special case of Eq. \eqref{keypr} in which the integral on the right side of \eqref{keypr} reduces to a sum. This enables us to understand the role of the spectrum  of the abelian algebra $\mathcal{O}_{\infty}$ in the decomposition \eqref{keypr}.  Thus, we assume that the set 
$\Xi^{\infty}$ is countable. Then
\begin{align}\label{keyotic}
 \mu_{\omega}(\cdot) = \int   \mu( \cdot | \nu ) \,d  P_{\omega}^{\infty}( \nu ) = \sum_{\nu \in \Xi^{\infty}} P^{\infty}_{\omega}(\nu) \mu( \cdot  | \nu), 
\end{align}
for some non-negative numbers $P^{\infty}_{\omega}(\nu)$, with 
$\sum_{\nu \in \Xi^{\infty}} P^{\infty}_{\omega}(\nu) = 1$. 
In the example described in Sect. \ref{SE} the set $\Xi^{\infty}$ is the spectrum of the operator 
$\mathcal{N}_P$ counting the number of electrons bound by $P$.

%Then Formula \eqref{ok} would have an infinite series instead of a finite sum. Moreover, from the pedagogical point of view, this example is very convenient to present concepts that in the general framework become too abstract and very difficult to handle.   

We recall that the spectrum of an abelian algebra is the set of non-trivial algebra-homomorphisms with values in the complex numbers; these homomorphisms are called ``characters''. We propose to characterize the characters of  $\mathcal{O}_{\infty}$, which can be identified with points, $\nu$, in the set $\Xi^{\infty}$:  \\
Let $f$ be an element of $\mathcal{L}_{\infty}$. Since the measures  
$\{ \mu( \cdot | \nu)  \}_{\nu \in \Xi^{\infty}}$ are $\Sigma_{\infty}$-ergodic and mutually singular, there are constants $ f_{\nu}$, $\nu \in \Xi^{\infty}$ such that 
\begin{equation}
f(\underline{\xi}) = \sum_{\nu \in \Xi^{\infty}} f_{\nu} \chi_{p^{-1}(\nu)}(\underline{\xi}), 
\end{equation}
where $p^{-1}(\nu)$ is the molecule in $\widetilde{\Xi}$ projecting onto $\nu \in \Xi^{\infty}$. 
We may then identify $f$ with the set $(f_{\nu})_{\nu \in \Xi^{\infty}}$. Thus,  $ \mathcal{O}_{\infty} = \Phi_{\mathcal{I}}(\mathcal{L}_{\infty})$ is isomorphic to the space, $C(\Xi^{\infty})$, of bounded continuous functions on the space $\Xi^{\infty}$, which is the spectrum of $\mathcal{O}_{\infty}$. Thus, the points 
$\nu \in \Xi^{\infty}$ appearing in Eq. \eqref{keyotic} are points in the spectrum of the abelian algebra 
$\mathcal{O}_{\infty}$.

\appendix

\section{Some basic definitions }\label{appendix}

In this section, we first recall some general notions and results from the theory of operator algebras. 

\begin{itemize}
\item[A.]{An algebra $\mathcal{E}$ of bounded linear operators is a $C^{*}$-algebra with identity if 
\begin{itemize}
\item{it is closed under taking the adjoint, $A \mapsto A^{*}$, of operators $A \in \mathcal{E}$, where $^{*}$ is an antilinear (anti-)involution on $\mathcal{E}$;}
\item{it is equipped with a norm, $\Vert \cdot \Vert$, with the property that
$$\Vert A^{*}A \Vert= \Vert A \Vert^{2}, \quad \forall A  \in \mathcal{E};$$}
\item{it is closed in the norm $\Vert \cdot \Vert$;}
\item{$\mathcal{E}$ contains an operator ${\bf{1}}$ with the property that
$$A\cdot {\bf{1}} = {\bf{1}}\cdot A = A, \quad \forall A \in \mathcal{E}.$$}
\end{itemize}
}
\item[B.]{The Gel'fand-Naimark-Segal (GNS) construction: Let $\pi$ be a representation of a $C^{*}$-algebra 
$\mathcal{E}$ on a Hilbert space $\mathcal{H}$, and let $\Psi$ be a unit vector in $\mathcal{H}$. Then
\begin{equation}\label{vectstate}
A \mapsto \psi(A):=\langle \Psi, \pi(A) \Psi \rangle_{\mathcal{H}}, \quad A\in \mathcal{E}
\end{equation}
is a state on $\mathcal{E}$; (see Sect. \ref{GF}).\\
Actually, \textit{every} state $\omega$ on $\mathcal{E}$ comes from a vector, $\Omega$, in a Hilbert space, 
$\mathcal{H}_{\omega}$, that carries a distinguished $^{*}$representation, $\pi_{\omega}$, of $\mathcal{E}$, as in \eqref{vectstate}:
\begin{equation}
\omega(A)=\langle \Omega, \pi_{\omega}(A) \Omega \rangle_{\mathcal{H}_{\omega}}, \quad \forall A \in \mathcal{E} . 
\end{equation}
This is the contents of the GNS construction.
The mathematical derivation of this construction is simple enough that we do not repeat it here; but see, e.g., \cite{Ta}, or Wikipedia.\\
If $\omega$ is a pure (i.e., extremal) state on $\mathcal{E}$ then the representation $\pi_{\omega}$ is irreducible.
For every $C^{*}$-algebra $\mathcal{E}$ there exists a family of pure states on $\mathcal{E}$ with the property that the direct sum of the GNS representations corresponding to these pure states is faithful, i.e., it distinguishes different elements of $\mathcal{E}$.
}
\item[C.]{Let $\mathcal{E}$ be a $^{*}$algebra of bounded linear operators acting on a Hilbert space 
$\mathcal{H}$ and containing the identity operator ${\bf{1}}$ on $\mathcal{H}$. We define the commutant,
$\mathcal{E}'$, of $\mathcal{E}$ to be the $^{*}$-algebra of all bounded linear operators, $B$, acting on $\mathcal{H}$ with the property that $B$ commutes with all operators in $\mathcal{E}$, i.e.,
$$[B,A]=0, \qquad \forall A \in \mathcal{E}.$$
The algebra $\mathcal{E}$ is a von Neumann algebra off it is equal to its double commutant, i.e., off
$$\mathcal{E}=(\mathcal{E}')'.$$
By von Neumann's double commutant theorem, this is equivalent to saying that $\mathcal{E}$ is closed in the topology determined by weak convergence of nets of bounded linear operators acting on $\mathcal{H}$.\\
Every von Neumann algebra is a $C^{*}$-algebra.
}
\item[D.]{A state $\omega$ on a von Neumann algebra $\mathcal{E}$ is called \textit{normal} iff it is continuous in the topology determined by weak convergence of nets of operators in $\mathcal{E}$.
}
\item[E.]{The \textit{center}, $\mathcal{Z}_{\mathcal{E}}$, of a $C^{*}$- or a von Neumann algebra $\mathcal{E}$ consists of all operators $Z \in \mathcal{E}$ commuting with all operators in $\mathcal{E}$, i.e.,
$$\mathcal{Z}_{\mathcal{E}}=\lbrace Z \vert Z \in \mathcal{E}, [Z,A]=0, \forall A \in \mathcal{E} \rbrace.$$\\
A von Neumann algebra with trivial center is called a \textit{factor}.\\
The \textit{centralizer}, $\mathcal{C}_{\omega}$, of a state $\omega$ on a von Neumann algebra 
$\mathcal{E}$ consists of all operators, $C$, in $\mathcal{E}$ with the property that 
$$\omega([C,A])=0, \qquad \forall A \in \mathcal{E}.$$
}
\item[F.]{{ In a von-Neumann algebra the notions of weak and $\sigma$-weak convergence coincide on norm bounded sets.}}
\end{itemize}
This completes our brief summary of basic notions and results on operator algebras.

\end{document}